\documentclass[12pt,fleqn,leqno]{article}
\usepackage{amsmath,amssymb}
\usepackage{amsthm,natbib,color}
\usepackage[dvips]{graphicx}

\usepackage{setspace}

\allowdisplaybreaks

\numberwithin{equation}{section}

\makeatletter
\renewcommand{\section}{\@startsection{section}{1}{0pt}{20pt}{6pt}{\large\bf}}
\renewcommand{\@seccntformat}[1]{\csname the#1\endcsname.\ }

\def\footnoterule{\kern -3pt \hrule width 18 true cm \kern 2.6pt}

\def\ni{\noindent}
\def\vs{\vspace}
\def\hs{\hspace}

\def\EE{\mathsf E}
\def\PP{\mathsf P}
\def\cC{\mathcal C}
\def\cD{\mathcal D}
\def\cK{\mathcal K}
\def\I{\mathcal I}

\def\R{I\!\!R}
\def\L{I\!\!L}
\def\wt{\widetilde}

\def\e{\text{e}}

\newcommand{\eps}{\varepsilon}
\newcommand{\p}{\! +\! }
\newcommand{\m}{\! -\! }

\newtheorem{theorem}{Theorem}[section]

\newtheorem{prop}[theorem]{Proposition}

\newtheorem{remark}[theorem]{Remark}

\begin{document}

\title{Mean Reversion Trading with Sequential Deadlines and  Transaction Costs}
\author{Yerkin Kitapbayev\thanks{Questrom School of Business, Boston University, Boston MA 02215; email:\,\mbox{yerkin@bu.edu}. Corresponding author.
} \and Tim Leung\thanks{Applied Mathematics Department,  University of Washington, Seattle WA 98195; email:\,\mbox{timleung@uw.edu}.} }
\date{\today }

\maketitle
\abstract{We study  the optimal timing strategies for trading a mean-reverting price process  with  a finite deadline to enter and  a  separate finite deadline to exit the market.  The price  process  is modeled   by a   diffusion with an affine drift that encapsulates a number of well-known   models, \mbox{including} the Ornstein-Uhlenbeck (OU) model, Cox-Ingersoll-Ross (CIR) model, Jacobi model, and   inhomogeneous geometric Brownian motion (IGBM) model. We analyze three types of  \mbox{trading} strategies: (i) the long-short  (long to open, short to close) strategy; (ii) the short-long  (short to open, long to close) strategy, and (iii) the chooser strategy whereby the trader has the added flexibility to  enter the market by taking either a long or short position, and subsequently close the position. For each strategy, we solve an optimal double stopping problem with sequential deadlines, and determine the optimal timing of trades. Our solution methodology utilizes the local time-space calculus of  \cite{Peskir2005a}  to derive  nonlinear integral equations of Volterra-type that uniquely characterize the trading boundaries. Numerical implementation of the integral equations provides examples of the optimal trading boundaries.}


{\par \leftskip=2.6cm \rightskip=2.6cm \footnotesize
    \par}

\vspace{10pt}

\begin{small}

\noindent {{JEL Classification: C41, G11, G12}}

\noindent {{MSC2010:} Primary 91G20, 60G40. Secondary 60J60, 35R35, 45G10.}

\noindent {{Keywords:} spread trading, optimal stopping, mean reversion, free-boundary problem, local time}
\end{small}

\vs{-18pt}

\vs{-18pt}


\newpage
\section{Introduction}
A major class of trading strategies in various markets, including equity, currency, fixed income, and futures markets,  is based on taking advantage of the mean-reverting behavior of asset prices. The mean-reverting price process can be that of a a single asset or derivative, or a portfolio of simultaneous positions in two or more highly correlated or co-moving assets, like stocks or exchange-traded funds (ETFs), or derivatives, such as futures and swaps.    Some practical examples of mean-reverting prices can be found in a number of empirical studies, including pairs of stocks and/or ETFs \citep{gatev2006pairs,Avellaneda2010,Montana2011,meanreversionbook2016}, divergence between futures and its spot \citep{futuresBS,futuresDaiKwok}, and spreads between physical commodity and  commodity stocks/ETFs \citep{kanamura2010profit,GoldMinerSpreads2013}. There are also automated approaches for identifying mean-reverting portfolios \citep{d2011identifying}.

In this paper, we investigate the problem of  optimal timing to trade  a mean-reverting price process with finite deadlines to enter and   subsequently exit the market. We consider a general  mean-reverting diffusion with an affine drift that encapsulates a number of well-known  and important models, including  the Ornstein-Uhlenbeck (OU) model \citep{Ornstein1930}, Cox-Ingersoll-Ross (CIR) model \citep{CIR85}, 
inhomogeneous geometric Brownian motion (IGBM) model  \citep{ZhaoIGBM,PhamIGBM}, and Jacobi model \citep{AFP2016}.

We analyze three types  of  trading strategies: (i) the long-short strategy; (ii) the short-long strategy, and (iii) the chooser strategy strategy in Sections 3, 4, and 5, respectively.  With the long-short (respectively, short-long) strategy,  the trader commits to first establish a long (respectively, short) position and close it out with an opposite position. In addition, we recognize  that there is a \textit{chooser option} embedded in the trader's decision whereby the trader can enter the market by taking either a long or short position. We recognize that after the trader enters the market  (long/short), the ensuing exit rule must coincide with that from the  long-short/short-long strategy. This leads us to examine the associated entry timing  under this chooser strategy.  In addition, we incorporate a fixed transaction cost  in each trade under any strategy considered herein.

Another feature of our trading problems is the introduction of sequential  deadlines for entry and   exit. As a consequence,  the amount of time available for the trader to close an open position is fixed and independent of the time at which the position was established.   In our related study, \cite{KitaLeung},  similar trading problems have been analyzed but without transaction costs and  the trader faces  the \emph{same} deadline by which all trading decisions have to be made. As such, the sooner   the trader enters the market, the more   amount of time that is left to exit. In contrast, if the trader enters the market later, there will be less time left to exit. In particular, if the first trade is executed very near the deadline, then there is high chance that the trader will close out the trade by the deadline rather than at the optimal boundary. In contrast,  the trader in this paper has the same fixed time window to close out an open position regardless of the entry timing.  Moreover, without transaction costs,  the  trading problem under the  chooser strategy  in \cite{KitaLeung} is \emph{trivial} in the sense that it can be decomposed into two uncoupled single stopping problems as opposed to a  sequential/double stopping problem as in the current paper with transaction costs.  Furthermore, our current formulation and solution approach give more analytical and  computational   tractability so that  the trading problem can be fully solved even with the features of transaction costs and the option not to enter at all (see Section 5 below).  This particular case  but with  a single deadline was left unsolved as an open problem in Section 6 of \cite{KitaLeung}.

For each strategy type, we solve a finite-horizon optimal double  stopping problem to determine the optimal trading strategies. In contrast to some related mean reversion trading problems with an infinite trading horizon  \citep{zhang2008trading,Ekstrom2011,LeungLi2014OU,LeungLiWang2014CIR,LeungLiWang2014XOU,ZhangexpOU}, closed-form solutions cannot be expected. Hence, we provide analytical representations that uniquely characterize  the value functions and optimal strategies corresponding to the double stopping problems. Specifically, our solution approach employs the local time-space calculus of \cite{Peskir2005a}  to derive  nonlinear integral equations of Volterra-type that uniquely characterize the trading boundaries. The use of  local time-space calculus of \cite{Peskir2005a} has also been applied previously in trading and option pricing applications in the context of multiple stopping problems \citep{Kita1,KitaLeung}.

To our best knowledge, the resulting representations for the value functions and optimal entry/exit boundaries, as summarized in Theorems \ref{th:2},   \ref{th:4}, and \ref{th:ch}, are new.  Furthermore, we discuss the implementation of the integral equations  as well as the resulting numerical examples.  In particular, we demonstrate that the value function increases as the trader is given more time to wait to enter the market, but the incremental value diminishes for longer  trading horizon (Figure 3).   We also numerically find  that the chooser option leads to an expansion of the trader's continuation (waiting) region compared with long-short or short-long strategies, which means that with the chooser option it is optimal to delay market entry whether the first position is long or short (Figure 5).

Our paper contributes to the growing literature on finite-horizon mean reversion trading.    \cite{elliott2005pairs} model the market entry timing by the first passage time of an Ornstein-Uhlenbeck (OU) process and specify  a fixed future date for market exit.  \cite{song2009stochastic}  present   a numerical stochastic approximation scheme to determine  the optimal buy-low-sell-high strategy over a finite horizon.  On trading futures on a mean-reverting spot,  \cite{LeungLiLiZheng2015} and \cite{JiaoLi2016}   discuss an approach that involves numerically solving a series of   variational inequalities using finite-difference methods.    Recognizing   the converging spread between   futures and spot, \cite{futuresDaiKwok}  propose a Brownian bridge model to find  the optimal timing to capture the spread. They also incorporate the chooser option embedded in the trading problem.
 \vs{6pt}


\section{Problem overview}

We fix a finite trading horizon $[0,T]$, and  filtered probability space $(\Omega, \mathcal{F} , (\mathcal{F}_t), \PP)$, where $\PP$ is the subjective probability measure representing the belief  of the trader, and $(\mathcal{F}_t)_{0\le t\le T}$ is the filtration  generated by the standard Brownian motion $(B_t)_{t\ge 0}$.  Let us define the corresponding state space $\I=(a,b)$ for the underlying asset price process $X$ where both $a<b$ may be both finite or infinite. We also allow $\I$ to be either closed/open or semi-closed/open interval.
We  model the asset price $X$ by the  mean-reverting diffusion process
\begin{align}\label{OU} 
dX_t=\mu(\theta\m X_t)dt+\sigma(X_t) dB_t, \quad X_0=x\in\I,
\end{align}
 where the constant parameters $\mu>0$ and $\theta\in \I$, represent the speed of mean reversion and long-run mean of the process, respectively. The diffusion coefficient $\sigma(x)$ is locally Lipchitz and satisfies the condition of linear growth. The latter guarantees the existence and uniqueness of the strong solution to \eqref{OU}. This framework contains three important models: the Ornstein-Uhlenbeck (OU) model when $\sigma(x)\equiv\sigma$ and $\I=\R$,
the  Cox-Ingersoll-Ross (CIR) model when $\sigma(x)=\sigma\sqrt{x}$ and $\I=\R^+$, the inhomogeneous geometric Brownian motion (IGBM) model when $\sigma(x) = \sigma x$ and $\I=\R^+$,
and the so-called Jacobi model when $\sigma(x)=\sigma \sqrt{(x-a)(b-x)}/(\sqrt{b}-\sqrt{a})$ and $\I=(a,b)$.

At any fixed time $u\ge t \ge 0$, we denote the random variable   $X^{t,x}_u$ as  the value of $X_u$ given $X_t=x$, and it has some probability density function
$p(\wt{x};u,x,t)$ for $\wt{x}$ and $x\in \I$. The mean function   of $X^{t,x}_u$ is  denoted by $m(u\m t,x)=\EE[X^{t,x}_u]$, where $\EE$ is the expectation under $\PP$.
It is well known that
\begin{align}\label{m} 
m(u\m t,x)&=x\e^{-\mu (u-t)}+\theta(1\m \e^{-\mu (u-t)}).
\end{align}
 The infinitesimal operator of $X$ is given as
\begin{align}\label{OU-operator} 
\L_X F(x)=\mu(\theta\m x)F'(x)+\frac{\sigma^2(x)}{2} F''(x)
\end{align}
for $x\in \I$ and any function $F\in C^2 (\I).$
\vs{2pt}

While observing  the mean-reverting price process $X$, the trader has the timing option to  enter the market by establishing a position in $X,$ but she must do so at or before a given  deadline  $T>0.$ This finite horizon can be as short as minutes and hours, or as long as days and weeks, depending on the trader type and the market that $X$ belongs.  Should the trader find it optimal not to enter the market by $T$, she can wait past time $T$ without an open position, and consider  the trading problem afresh afterwards.   If the trader establishes a position in $X$ at time $\tau \le T,$ then she needs to subsequently close the position within a fixed time window of length $T'>0$. More precisely, she must liquidate by time $\tau+T'.$ In summary, the trader faces a sequence of two finite-horizon optimal timing problems, for which she needs to determine the respective optimal times to  enter and exit the market.

We analyze three trading  strategies: (i) the
\textit{long-short} strategy, whereby the trader takes long position in the spread first and later reverses the position to close (Section \ref{ls}); (ii) \textit{short-long} strategy, whereby the trader shorts the spread to start and then  close by taking the opposite (long) position  (Section \ref{sl}), and (iii)
 the \textit{chooser} strategy, i.e. the trader  can   take either a  long or short position in the spread when entering the market  (Section \ref{o}), and subsequently liquidates by taking the opposite position. Throughout,  the trader is assumed to be  \emph{risk-neutral} and seek to maximize expected discounted linear payoff under $\PP$.

 In all our trading problems, we  also incorporate  fixed transaction costs for entry and exit, so it only makes sense for the trader to open a position if the expected profit from the trade exceeds the fees.  Otherwise, the trader will find it optimal not to enter at all. This also allows us to examine the effects of transaction cost on the  value of the mean reversion trading opportunity.

\section{Optimal long-short strategy}\label{ls}

 In this section we consider the long-short strategy, whereby the trader, who seeks to enter the market by taking a long position at any time $\tau$ before the deadline $T$  and subsequently close  it  within a time window of length $T'$.
We formulate this problem sequentially by  backward induction. Suppose that the trader has already established a long position in $X$. In order to find the best time to liquidate  the position before $T'$, the trader solves the optimal stopping problem
 \begin{equation} \label{tc-problem-1}  
V^{1,L}(x)=\sup \limits_{0\le\zeta\le T'}\EE \left[\e^{-r\zeta} (X^x_\zeta-c)\right],
 \end{equation}
 for $x\in\I$  where $\zeta$ is a stopping time w.r.t. the filtration $(\mathcal{F}_t)$ and $c>0$ is the fixed transaction cost.  The value function $V^{1,L}$ represents the expected value of a long position in $X$  with the timing option to liquidate before the  deadline $T'$.

 Tracing backward in time, in order to establish this long position, the trader would need to pay the prevailing value of $X$ plus the transaction cost $c$.   Thus at entry time $\tau$, the trader pays $X^x_\tau+c$ and obtain the long position with the maximized expected value $V^{1,L}(X^{t,x}_\tau)$. Therefore, the difference $(V^{1,L}(X^{t,x}_\tau) - X^{t,x}_\tau - c)$ can be  viewed as the reward of the trader received upon entry. Naturally, the trader should only enter only if this difference is  strictly positive; otherwise, the trader has the right not to enter at all by the deadline $T$. Hence, the trader's optimal entry timing problem is given by
  \begin{equation} \label{tc-problem-2}  
V^{1,E}(t,x)=\sup \limits_{t\le\tau\le T}\EE \left[\e^{-r(\tau-t)}(V^{1,L}(X^{t,x}_\tau)\m X^{t,x}_\tau\m c)^+\right],
 \end{equation}
 for $(t,x)\in[0,T]\times \I$ where $\tau$ is a stopping time w.r.t. the filtration $(\mathcal{F}_t)$. Should the trader choose to enter the market, she must do so at or before the deadline $T$.

We first discuss the solution to the optimal exit problem \eqref{tc-problem-1} in Section~\ref{ls-exit}  and then solve the optimal entry  problem  \eqref{tc-problem-2} in Section~\ref{ls-entry}.

 \subsection{Optimal exit problem}\label{ls-exit}
We now solve the optimal exit timing  problem in \eqref{tc-problem-1}. To this end, let us define the time-dependent version of $V^{1,L}(x)$ by
 \begin{equation} \label{tc-problem-1-a}  
V^{1,L}(t,x;T')=\sup \limits_{t\le\zeta\le T'}\EE \left[\e^{-r(\zeta-t)} (X^{t,x}_\zeta-c)\right],
 \end{equation}
for $t\in[0,T')$ and $x\in\I$ so that $V^{1,L}(x)=V^{1,L}(0,x;T')$.
This problem has been already solved for the OU process in Section 3.1 of \cite{KitaLeung},  though we briefly extend the result for other mean-reverting processes described in \eqref{OU}.
Let us first apply Ito's formula along with the optional sampling theorem to get
\begin{align} \label{Ito-tc}  
\EE \left[\e^{-r(\zeta-t)} (X^{t,x}_\zeta-c)\right]=x-c+\EE \left[\int_t^\tau \e^{-r(s-t)}H^{1,L}(X^{t,x}_s)ds\right],
 \end{align}
 for $t\in[0,T')$ and $x\in\I$ and where the function $H^{1,L}$ is an affine function in $x$ defined by
 \begin{align} \label{H-1}  H^{1,L}(x)=-(\mu\p r)x+\mu\theta+rc.
\end{align}
We also define the unique root of function $H^{1,L}$ by
\begin{align}  
x^*:=&(\mu\theta+rc)/(\mu+r).\label{xstar1}\end{align}
We note that it is not guaranteed that $x^*\in \I  = (a,b)$. When $x^*$ falls outside of the state space, we summarize the scenarios as follows.

\begin{prop} If $x^*\le a$, then $H^{1,L}$ is always negative on $\I$ and it is optimal to liquidate the position at $t=0$.
 If $x^*\ge b$,  then $H^{1,L}$ is always positive on $\I$ and it is optimal to wait until the end and exit from the position at $t=T'$.

 \end{prop}

 \begin{remark}
 We note that for usual values of $r$ and $c$, the root $x^*$ is near $\theta$, which belongs to $\I$. Moreover,  this proposition is not relevant to the OU model as $\I=\R$ and trivial cases are automatically excluded.

 We also note that one scenario that can lead to $x^*\ge b$ is when the  transaction cost $c$ is very high (see \eqref{xstar1}). Then it makes economic  sense to   exit at the deadline $T'$ since  the trader can delay the   payment of the large fee. Secondly, if the risk-free rate $r$ is very large and $a$ is positive finite, then $x^*$ tends to $c$. If additionally   $c<a$, then it is financially intuitive to  immediately exit, as indicated  by Proposition 3.1(i), since the  trader would prefer to  obtain the cash now and invest it to the bank account at a high interest rate.

 \end{remark}

 Hence, to exclude two trivial cases, in what follows we assume that $x^*\in(a,b)$.
 Let us now define
 the   function
\begin{align} \label{L} 
\cK^{1,L}(t,u,x,z):=&-\e^{-r(u-t)}\EE \big[H^{1,L}(X^{t,x}_u) 1_{\{X^{t,x}_u \ge z\}}\big]\\=&-\e^{-r(u-t)}\int_{z}^{b\wedge\infty} H^{1,L}(\wt{x})\,p(\wt{x};u,x,t)\,d\wt{x},\notag
 \end{align}
 for $u\ge t>0$ and $x,z\in\I$, where $p$ is the probability density function of $X^{t,x}_u$  and
The theorem below summarizes the solution to the optimal exit problem.

\begin{theorem}\label{th:1}
The optimal stopping time  for \eqref{tc-problem-1-a} is given   by
\begin{align} \label{OST-2}  
\zeta_*^{1,L}=\inf\ \{\, t\leq s\leq T': X^{t,x}_{s}\ge b^{1,L}(s) \,\},
 \end{align}
 where the function   $b^{1,L}(\cdot)$ is the optimal exit boundary  that  can be characterized as the unique solution to a nonlinear integral equation of Volterra type
\begin{align}\label{th-1-2}  
b^{1,L}(t)\m c=\e^{-r(T-t)}m(T'\m t,b^{1,L}(t)) +\int_t^{T'} \cK^{1,L}(t,u,t,b^{1,L}(t),b^{1,L}(u))du,
\end{align}
for $t\in[0,T']$ in the class of continuous decreasing functions $t\mapsto b^{1,L}(t)$ with
\begin{equation}\label{b1LT} 
b^{1,L}(T')=x^*.
\end{equation}
The value function $V^{1,L}$ in \eqref{tc-problem-1-a} admits the  representation
\begin{align}\label{th-1-1}  
V^{1,L}(t,x;T')=\e^{-r(T'-t)}m(T'\m t,x) +\int_t^{T'}\cK^{1,L}(t,u,x,b^{1,L}(t\p u))du,
\end{align}
for $t\in[0,T']$ and $x\in \I$.
\end{theorem}

\begin{figure}[t]
\begin{center}
\includegraphics[scale=0.7]{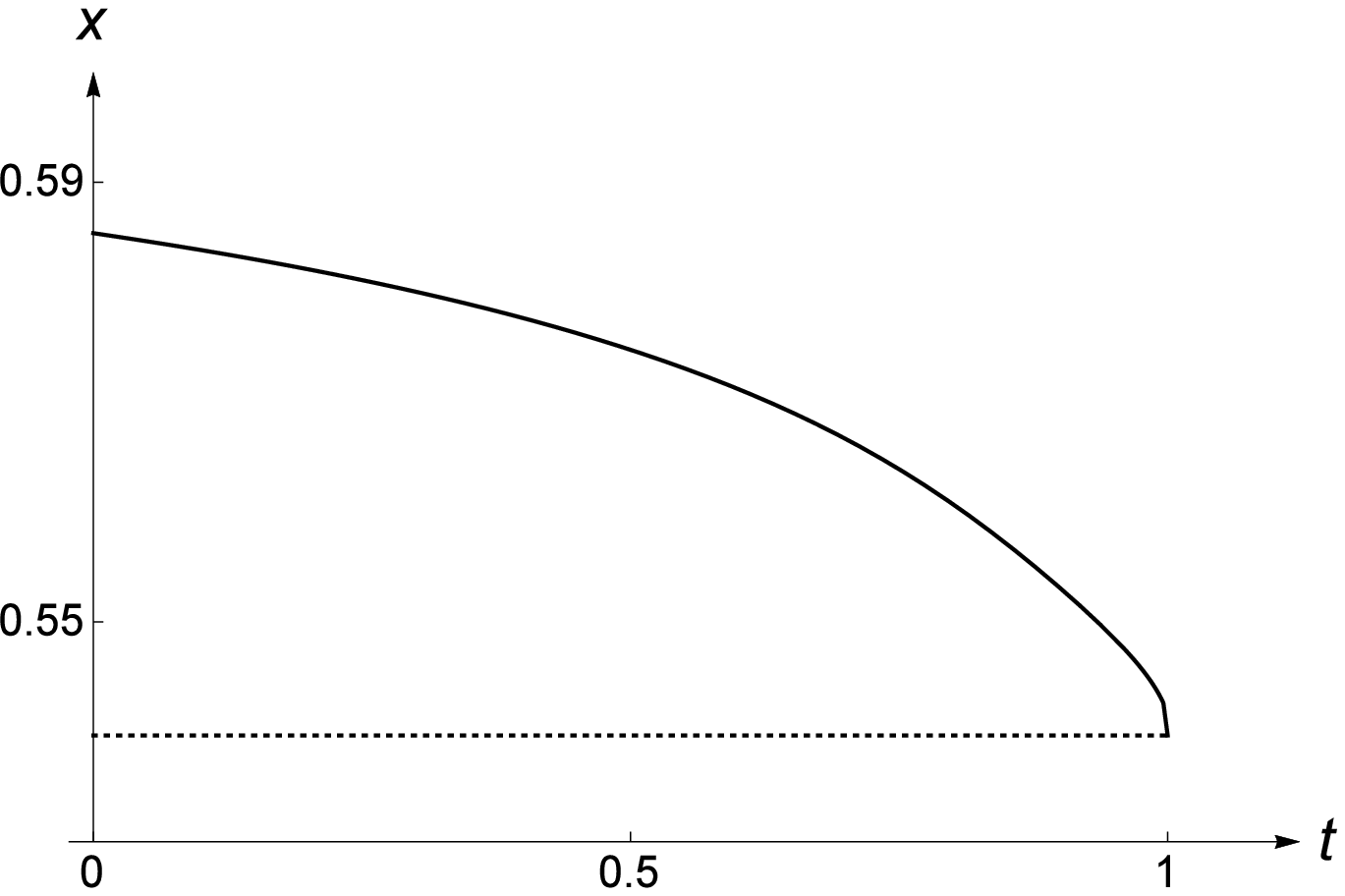}
\end{center}
{\par \leftskip=1.6cm \rightskip=1.6cm \small \ni  
\textbf{Figure 1.} The optimal exit boundary $b^{1,L}(t)$, $t\in[0,T'],$
corresponding to  the long-short strategy problem computed by numerically solving the  integral equation \eqref{th-1-2} under the OU model. The parameters are:  $T'=1$ year, $r=0.01$, $c=0.01$, $\theta=0.54$, $\mu=16$, $\sigma=0.16$, $b^{1,L}(T')=x^*=0.539669.$  A time discretization with  $500$ steps for the interval $[0,T']$ is used.
\par}
\end{figure}

 We also recall that $V^{1,L}(t,x)$ and $b^{1,L}(t)$ solve the following free-boundary problem in $[0,T']\times\I$:
 \begin{align} \label{PDE} 
 &V^{1,L}_t \p\L_X V^{1,L} \m rV^{1,L}=0  &\text{for}\;  x<b^{1,L}(t),\\
\label{IS}&V^{1,L}(t,b^{1,L}(t))=b^{1,L}(t)-c &\text{for}\; t\in[0,T'),\\
\label{SP}&V^{1,L}_x (t,b(t))=1 &\text{for}\; t\in[0,T'), \\
\label{FBP1}&V^{1,L}(t,x)>x-c &\text{for}\;  x<b^{1,L}(t),\\
\label{FBP2}&V^{1,L}(t,x)=x-c &\text{for}\;  x\ge b^{1,L}(t).
\end{align}
By Theorem 1 of \cite{KitaLeung}, we have the following properties:
\begin{align} \label{Prop-1} 
&V^{1,L}\;\text{is continuous on}\; [0,T']\times\I,\\
\label{Prop-3}&x\mapsto V^{1,L}(t,x)\;\text{is increasing and convex on $\I$ for each $t\in[0,T']$},\\
\label{Prop-4}&t\mapsto V^{1,L}(t,x)\;\text{is decreasing on $[0,T']$ for each $x\in \I$},\\
\label{Prop-5}&t\mapsto b^{1,L}(t)\;\text{is decreasing and continuous on $[0,T']$ with $b^{1,T}=x^*$.}
\end{align}

  Figure 1 displays the optimal exit boundary $b^{1,L}(t)$, over $t\in[0,T']$,   obtained by numerically solving the  integral equation \eqref{th-1-2} under the OU model (i.e. $\sigma(x) \equiv \sigma$). As expected from \eqref{Prop-5}, the boundary is decreasing and continuous on $[0,T']$, with limit $b^{1,L}(T')=x^*=0.539669$ (see \eqref{xstar1}).  Interestingly, this limit depends on the long-run mean $\theta$, and speed of mean reversion $\mu$, along with the interest rate $r$ and transaction cost $c$, but it does not depend on the volatility parameter $\sigma$. In fact, since  Theorem \ref{th:1} holds for a general volatility function $\sigma(x)$, this means that   mean-reverting models with the same values for the parameters $(\theta, \mu, r, c)$ but different specifications for $\sigma(x)$ will still  have the same limit given in \eqref{b1LT} at $T'$. This is consequence of the linear payoff so that the function $H^{1,L}$ does not depend on $\sigma(x)$.

 \subsection{Optimal entry problem}\label{ls-entry}

Having solved  for the value function $V^{1,L}(t,x;T')$  for the optimal exit timing problem, we   obtain  in particular
\begin{equation}\label{payoff-1-L}  
V^{1,L}(x)=V^{1,L}(0,x;T'), \qquad x\in\I,
\end{equation}
using \eqref{th-1-1}.  In turn, $V^{1,L}(x)$ is an input to the optimal entry problem. Indeed, we recall  \eqref{tc-problem-2}, and denote the entry payoff function by
 \begin{align} 
 G^{1,E}(x)=(V^{1,L}(x)-x-c)^+, \qquad x\in\I.
 \end{align}
 Then, the optimal timing to enter the market is found from the following optimal stopping problem
 \begin{equation} \label{tc-problem-3}  
V^{1,E}(t,x)=\sup \limits_{t\le\tau\le T}\EE \left[\e^{-r(\tau-t)}G^{1,E}(X^{t,x}_\tau)\right],
 \end{equation}
where the supremum is taken over all $(\mathcal{F}_t)$- stopping times $\tau\in[t,T]$. We define the threshold $\gamma^{1,L}$ as the unique root of the equation $V^{1,L}(x)-x-c=0$ so that using the convexity of $V^{1,L}$ we can rewrite
\begin{align}\label{G1E}  
G^{1,E}(x)=(V^{1,L}(x)-x-c)1_{\{x<\gamma^{1,L}\}},
\end{align}
for every $x\in\I$. In other words, it is not rational to enter into position when $X>\gamma^{1,L}$ as the expected profit is negative and thus the strategy of not trading at all strictly dominates the immediate entry.
\vs{6pt}

 Let us define the continuation and entry regions, respectively,
\begin{align} \label{C-2}  
&\cC^{1,E}= \{\, (t,x)\in[0,T)\! \times\! \I:V^{1,E}(t,x)>G^{1,E}(x)\, \} ,\\
 \label{D-2}&\cD^{1,E}= \{\, (t,x)\in[0,T)\! \times\! \I:V^{1,E}(t,x)=G^{1,E}(x)\, \}.
 \end{align}
Then the optimal entry time in \eqref{tc-problem-3} is given by
\begin{align} \label{OST-2}  
\tau^{1,E}_*=\inf\ \{\ t\leq s\leq T:(s,X^{t,x}_{s})\in\cD^{1,E}\ \}.
 \end{align}

  We employ  the local time-space formula on curves  \citep{Peskir2005a}) for $e^{-r(s-t)}G^{1,E}(X^{t,x}_s)$, the optional sampling theorem and the smooth-fit property \eqref{SP} at $b^{1,L}(0)$ to obtain
\begin{align} \label{Ito-tc}  
 \EE \left[\e^{-r(\tau-t)}G^{1,E}(X^{t,x}_\tau)\right]=&\;G^{1,E}(x)+\EE \left[\int_t^\tau \e^{-r(s-t)}H^{1,E}(X^{t,x}_s)ds\right]\\
 &+\EE \left[\int_t^\tau \e^{-r(s-t)}(1-V^{1,L}_x(X^{t,x}_s))d\ell^{\gamma^{1,L}}_s\right],\nonumber
 \end{align}
for $t\in[0,T)$, $x\in\I$ and arbitrary stopping time $\tau$ of process $X$. The function $H^{1,E}$ is defined as  $H^{1,E}(x):=(\L_X G^{1,E} \m rG^{1,E})(x)$ for $x\in\I$. By \eqref{G1E} and \eqref{PDE}, it is  given as
\begin{align} \label{H-tc}  
H^{1,E}(x)=((\mu\p r)x-\mu\theta+rc )1_{\{x<\gamma^{1,L}\}},
\end{align}
 for $x\in \I$. We define its root $x_*=(\mu\theta-rc)/(\mu+r)\le x^*$.  In \eqref{Ito-tc}, we have also introduced   the local time  $(\ell^{\gamma^{1,L}}_s)_{s\ge t}$ that the process $X$ spends at $\gamma^{1,L}$, namely,
\begin{align} \label{LT} 
 \ell^{\gamma^{1,L}}_s :=\PP-\lim_{\eps \downarrow 0}\frac{1}{2\eps}\int_t^{s} 1_{\{\gamma^{1,L}-\eps<X^{t,x}_u<\gamma^{1,L}+\eps\}}d\left \langle X \right \rangle_u.
 \end{align}

 To exclude the degenerate cases we provide the following proposition.

 \begin{prop}  If $x_*\wedge \gamma^{1,L}\le a$, then $H^{1,E}$ is always non-negative on $\I$ and it is optimal to wait until the end and enter at $t=T$.  If $x_*\wedge \gamma^{1,L}\ge b$,  then $H^{1,E}$ is always negative on $\I$ and local time term is zero, thus it is optimal to enter immediately at $t=0$. \end{prop}

 From now on we assume that $x_*\wedge \gamma^{1,L}\in(a,b)$.
In this case, the function $H^{1,E}$ is negative when $x<x_*\wedge \gamma^{1,L}$. It is not optimal to enter into the position when $x_*\wedge \gamma^{1,L}<X_s$
as $H^{1,E}$ is non-negative there and the local time term is always non-negative. Also, near $T$ it is optimal to enter at once when $X_s<x_*\wedge \gamma^{1,L}$ due to lack of time to compensate the negative $H^{1,E}$. This gives us the terminal condition of the exercise boundary  at $T$ (see \eqref{Prop-5-tc}).
We also note that if $a=-\infty$, the equation \eqref{Ito-tc} shows that the entry region is non-empty for all $t\in[0,T)$, as for large negative $x$ the integrand $H^{1,E}$
is very negative and thus it is optimal to enter immediately due to the presence of the finite deadline $T$.

 Since the payoff function $G^{1,E}$ and underlying process $X$ are time-homogenous, we have that the entry region $\cD^{1,E}$ is  right-connected.
Next, we show that  $\cD^{1,E}$ is down-connected.
Let us take $t>0$ and $x<y<x_*\wedge \gamma^{1,L}$ such that $(t,y)\in \cD^{1,E}$.
Then, by right-connectedness of the entry region, we have that $(s,y)\in \cD^{1,E}$ as well for any $s>t$. If we now run the process $(s,X_s)_{s\ge t}$ from $(t,x)$, we cannot hit the level $x_*\wedge \gamma^{1,L}$ before entry (as $x<y$), thus the local time term in \eqref{Ito-tc} is 0 and integrand $H^{1,E}$ is strictly negative before $\tau^{1,E}_*$.
Therefore it is optimal to entry at $(t,x)$ and we obtain down-connectedness of the entry region $\cD^{1,E}$.

Hence there exists an optimal entry boundary $b^{1,E}$ on $[0,T]$ and the corresponding stopping time 
\begin{align}   
&\tau=\inf\ \{\, t\leq s\leq T: X^{t,x}_{s}\le b^{1,E}(s)\, \},
 \end{align}
is optimal for the entry timing problem  \eqref{tc-problem-2} and $a\le b^{1,E}(t)<x_*\wedge \gamma^{1,L}$ for $t\in[0,T)$. Moreover, $b^{1,E}$ is increasing on $[0,T]$ and $b^{1,E}(T-)=x_*\wedge \gamma^{1,L}$.

The value function $V^{1,E}$ and boundary $b^{1,E}$ solve the following free-boundary problem:
\begin{align} \label{PDE-tc}  
&V^{1,E}_t \p\L_X V^{1,E} \m rV^{1,E}=0 &\hs{-30pt}\text{in}\;  \cC^{1,E},\\
\label{IS-tc}&V^{1,E}(t,b^{1,E}(t))=V^{1,L}(t,b^{1,E}(t))\m b^{1,E}(t)\m c &\hs{-30pt}\text{for}\; t\in[0,T),\\
\label{SF-tc}&V^{1,E}_x(t,b^{1,E}(t))=V^{1,L}_x(t,b^{1,E}(t))\m 1&\hs{-30pt}\text{for}\; t\in[0,T), \\
\label{FBP1-tc}&V^{1,E}(t,x)>G^{1,E}(x) &\hs{-30pt}\text{in}\; \cC^{1,E},\\
\label{FBP2-tc}&V^{1,E}(t,x)=G^{1,E}(x) &\hs{-30pt}\text{in}\; \cD^{1,E},
\end{align}
where the continuation set $\cC^{1,E}$ and the entry region $\cD^{1,E}$ are given by
\begin{align} \label{CC-tc}  
&\cC^{1,E}= \{\, (t,x)\in[0,T)\! \times\! \I: x>b^{1,E}(t)\, \}, \\ 
 \label{DD-tc}&\cD^{1,E}= \{\, (t,x)\in[0,T)\! \times\! \I:x\le b^{1,E}(t)\, \}.
 \end{align}
Using standard arguments, it follows  that
\begin{align} \label{Prop-1-tc}  
&V^{1,E}\;\text{is continuous on}\; [0,T]\times\I\,,\\
\label{Prop-3-tc}&x\mapsto V^{1,E}(t,x)\;\text{is convex on $\I$ for each $t\in[0,T]$}\,,\\
\label{Prop-4-tc}&t\mapsto V^{1,E}(t,x)\;\text{is decreasing on $[0,T]$ for each $x\in \I$}\,,\\
\label{Prop-5-tc}&t\mapsto b^{1,E}(t)\;\text{is continuous on $[0,T]$ with}\; b^{1,E}(T-)=x_*\wedge \gamma^{1,L}.
\end{align}

 Let us define the function $\cK^{1,E}$ as
\begin{align}
\label{L-tc}
\cK^{1,E}(t,u,x,z)&=-\e^{-r(u-t)}\EE \big[H^{1,E}( X^{t,x}_u) 1_{\{X^{t,x}_u \le z\}}\big]\\
&=-\e^{-r(u-t)}\int_{-\infty\vee a}^z H^{1,E}(\wt{x})\,p(\wt{x};u,x,t)\,d\wt{x},
 \end{align}
 for $u\ge t\ge 0$ and $x,z\in\I$. In particular, when $X$ is an  OU process, then the function $\cK^{1,E}$ can be efficiently rewritten in terms of standard normal cumulative and probability
 distribution functions. For more complicated distributions, the computation of    $\cK^{1,E}(t,u,x,z)$  in \eqref{L-tc} may require   numerical univariate integration.
  \vs{6pt}

Applying local time-space formula (\cite{Peskir2005a}) to the discounted process $\e^{-r(s-t)}V^{1,E}(s,X^{t,x}_s)$,  along with \eqref{PDE-tc}, the definition of $H^{1,E}$, and  the smooth-fit property \eqref{SF-tc}, we obtain 
 \begin{align} \label{LTSF-1}  
&\e^{-r(s-t)}V^{1,E}(s,X^{t,x}_s)\\
=\;&V^{1,E}(t,x)+M_s\nonumber\\
 &+ \int_t^{s}
\e^{-r(u-t)}\left(V^{1,E}_t \p\L_X V^{1,E}
\m rV^{1,E}\right)(u,X^{t,x}_u)
 I(X^{t,x}_u \neq b^{1,E}(u))du\nonumber\\
 &+\frac{1}{2}\int_s^{t}
\e^{-r(u-t)}\left(V^{1,E}_x (u,X^{t,x}_u +)-V^{1,E}_x (u,X^{t,x}_u -)\right)I\big(X^{t,x}_u=b^{1,E}(u)\big)d\ell^{b^{1,E}}_u\nonumber\\
  =\;&V^{1,E}(t,x)+M_s +\int_t^{s}
\e^{-r(u-t)}H^{1,E}(X^{t,x}_u) I(X^{t,x}_u \ge b^{1,E}(u))du\,,\nonumber
  \end{align}
for $s\in[t,T]$ where    $M=(M_s)_{s\ge t}$ is a martingale, and $(\ell^{b^{1,E}}_s)_{s\ge t}$ is the local time process of $X^x$ at the boundary $b^{1,E}$ given by
\begin{align} \label{LT-2}  
 \ell^{b^{1,E}}_s:=\PP-\lim_{\eps \downarrow 0}\frac{1}{2\eps}\int_t^{s} 1_{\{b^{1,E}(u)-\eps<X^{t,x}_u<b^{1,E}(u)+\eps\}}d\left \langle X\right \rangle_u.
 \end{align}
 We then   obtain the main result of this section.

\begin{theorem}\label{th:2}
The optimal entry boundary $b^{1,E}: [0,T]\mapsto \mathbb{R}$ can be characterized as the unique solution to the recursive integral equation
\begin{align}\label{th-1-tc}  
V^{1,L}(b^{1,E}(t))\m b^{1,E}(t)\m c=&e^{-r(T-t)}\EE [G^{1,E}(X^{t,b^{1,E}(t)}_{T})]\\
&+\int_t^{T} \cK^{1,E}(t,u,b^{1,E}(t),b^{1,E}(u))du\,,\notag
\end{align}
for $t\in[0,T]$ in the class of continuous increasing functions with $b^{1,E}(T-)=\gamma^{1,L}\wedge x_*$.
 The value function $V^{1,E}$ admits the representation
 \begin{align}\label{th-2-tc}  
V^{1,E}(t,x)=e^{-r(T-t)}\EE [G^{1,E}(X^{t,x}_{T})]+ \int_t^{T}\cK^{1,E}(t,u,x,b^{1,E}(u))du\,,
\end{align}
for $t\in[0,T]$ and $x\in \I$.
\end{theorem}

\begin{proof} The representation \eqref{th-2-tc} follows from    \eqref{LTSF-1}. Specifically, we set 
  $s=T$, take   expectations on both sides, and apply the optional sampling theorem for $M$, rearrange terms and recall \eqref{L-tc} and  the terminal condition
$V^{1,E}(T,x)=G^{1,E}(x)$ for all $x\in \I$. The integral equation \eqref{th-1-tc} is obtained by inserting $x=b^{1,E}(t)$ into \eqref{th-2-tc} and using \eqref{IS-tc}.
\end{proof}

\begin{figure}[h]
\begin{center}
\includegraphics[scale=0.6]{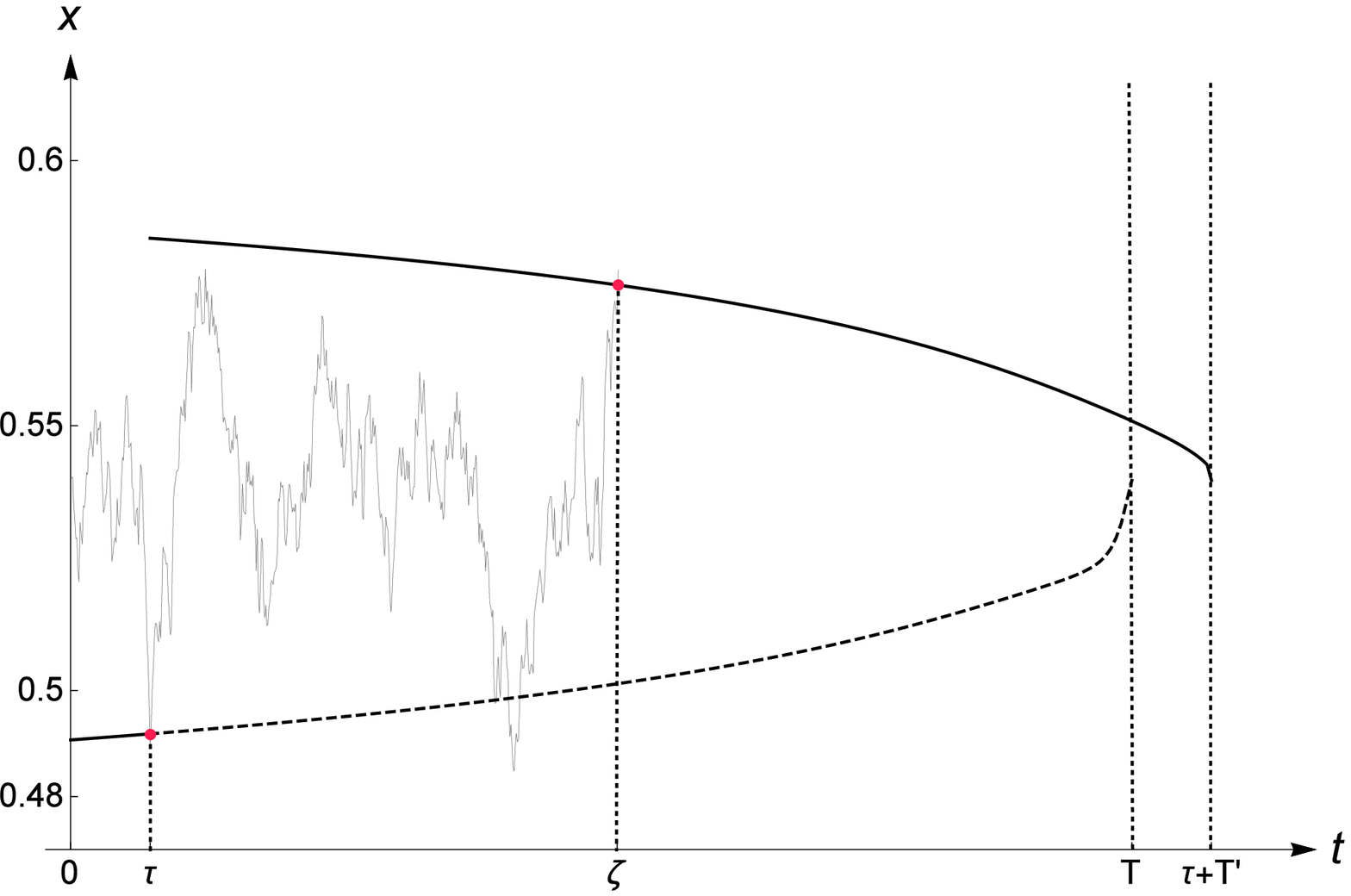}
\end{center}

{\par \leftskip=1.6cm \rightskip=1.6cm \small \ni \vs{-10pt}

\textbf{Figure 2.} A sample path of OU process, with the entry time $\tau=0.075$ when the path hits the optimal entry boundary $b^{1,E}$ (lower dashed), and the subsequent liquidation time $\zeta=0.516$ when the path reaches the optimal exit boundary
$b^{1,L}$ (upper solid).  The optimal entry boundary $b^{1,E}$ and exit boundary $b^{1,L}$
 corresponding to  the long-short strategy  are     computed from the  integral equations \eqref{th-1-tc}  and \eqref{th-1-2}  respectively.
The parameters are:  $T=T'=1$ year, $r=0.01$, $c=0.01$, $\theta=0.54$, $\mu=16$, $\sigma=0.16$, $b^{1,E}(T)=x_*=0.539656$, $\gamma^{1,L}=0.5545$.

\par} \vs{10pt}

\end{figure}

In Figure 2, we display a  sample path of the mean-reverting price process $X$  under the OU model, with parameters $\theta=0.54$, $\mu=16$, and $\sigma=0.16$. As we can see, the trader enters the market when $X$ reaches the lower boundary $b^{1,E}$ at a time $\tau\le T$. Subsequently, the trader waits for the process $X$ to hit the upper boundary $b^{1,L}$ to liquidate the (long) position  at time $\zeta$. The optimal entry and exit boundaries,  $b^{1,E}$ and  $b^{1,L}$, are     computed from the  integral equations \eqref{th-1-tc}  and \eqref{th-1-2}  respectively.  If the path $X$ does not  reach the lower boundary $b^{1,E}$ by time $T$, then the trader will simply leave the market without an open position.
  Also notice that the exit boundary is only relevant to the trader after market entry, and the trader has the time window of length $T'$ to close the position. If the price process $X$ fails to  touch the upper exit boundary  $b^{1,L}$ by the end of the time window, then the trader will be  forced to liquidate at the end.

To illustrate the value function's dependence on the deadlines to trader, we plot the map $T\mapsto V^{1,E}(0,\theta;T)$, which is the value function evaluated at $x= \theta$, as  function of the deadline $T$ for $T'=1$ (solid) and  $T'=0.5$ (dashed)  in Figure 3. When given more time to trade, whether to decide to enter or exit the market, the trader can achieve a higher expected value from trading $X$. Nevertheless, the fact that  $V^{1,E}(0,x;T)$  is  increasing concave in $T$  indicates a diminishing incremental value of a longer   horizon for the trader to decide to enter the market.

\begin{figure}[t]
\begin{center}
\includegraphics[scale=0.8]{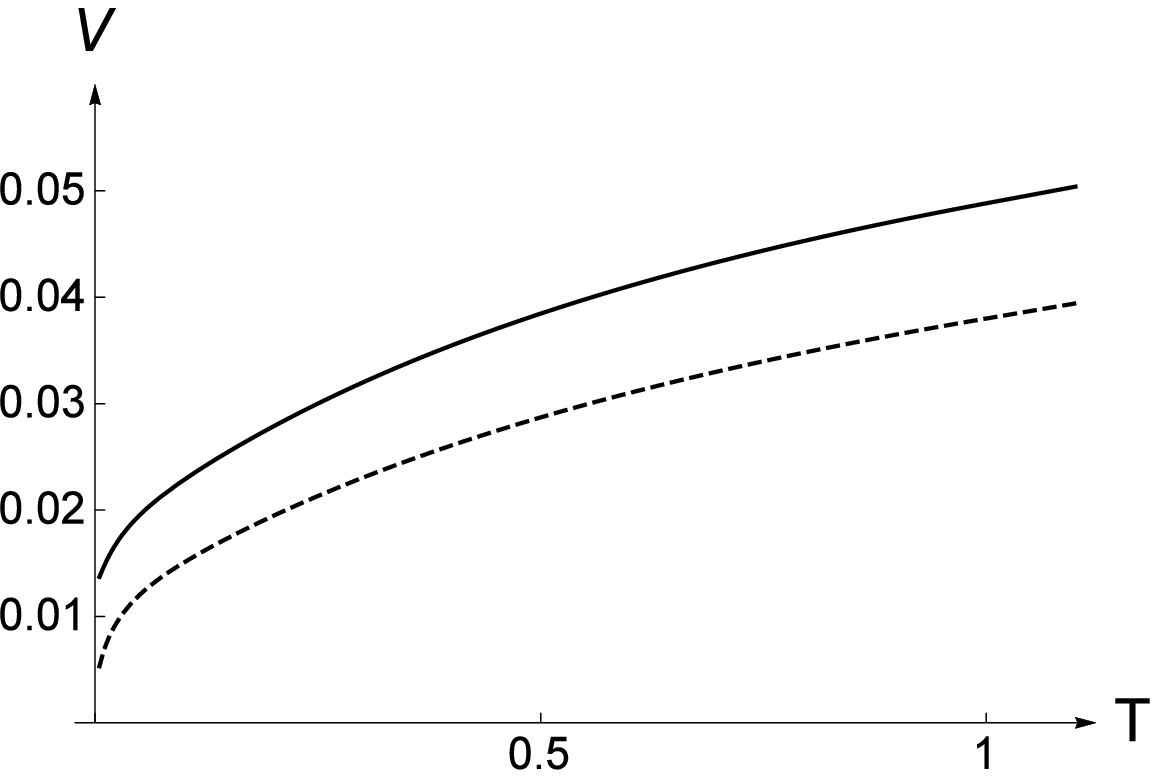}
\end{center}
{\par \leftskip=1.6cm \rightskip=1.6cm \small \ni  
\textbf{Figure 3.} The value functions $V^{1,E}(0,x;T)$  for the long-short strategy, evaluated at $x= \theta$ and plotted as  functions of the deadline $T$ for $T'=1$ year (solid) and  $T'=0.5$ year (dashed). The parameters are: $r=0.01, c=0.01, \theta=0.54, \mu=16, \zeta=0.16.$
\par} \end{figure}

Figure 4  shows that the  value function $x\mapsto V^{1,E}(0,x)$ (evaluated at time $t=0$) dominates the  payoff function $G^{1,E}(x)=(V^{1,L}(x)-x-c)^+$,  and they coincide for $x$ lower than $b^{1,E}(0)$. The value function is decreasing in $x$ but becomes flatter  for larger $x$. This suggests that being far above from the entry boundary at time $0$ does not significantly diminish the value of trading in $X$.  This can be explained by the mean-reverting property of $X$.The value function is also reduced when the transaction cost $c$ increases from $0.01$ to $0.02$.


\begin{figure}[h]
\begin{center}
\includegraphics[scale=0.7]{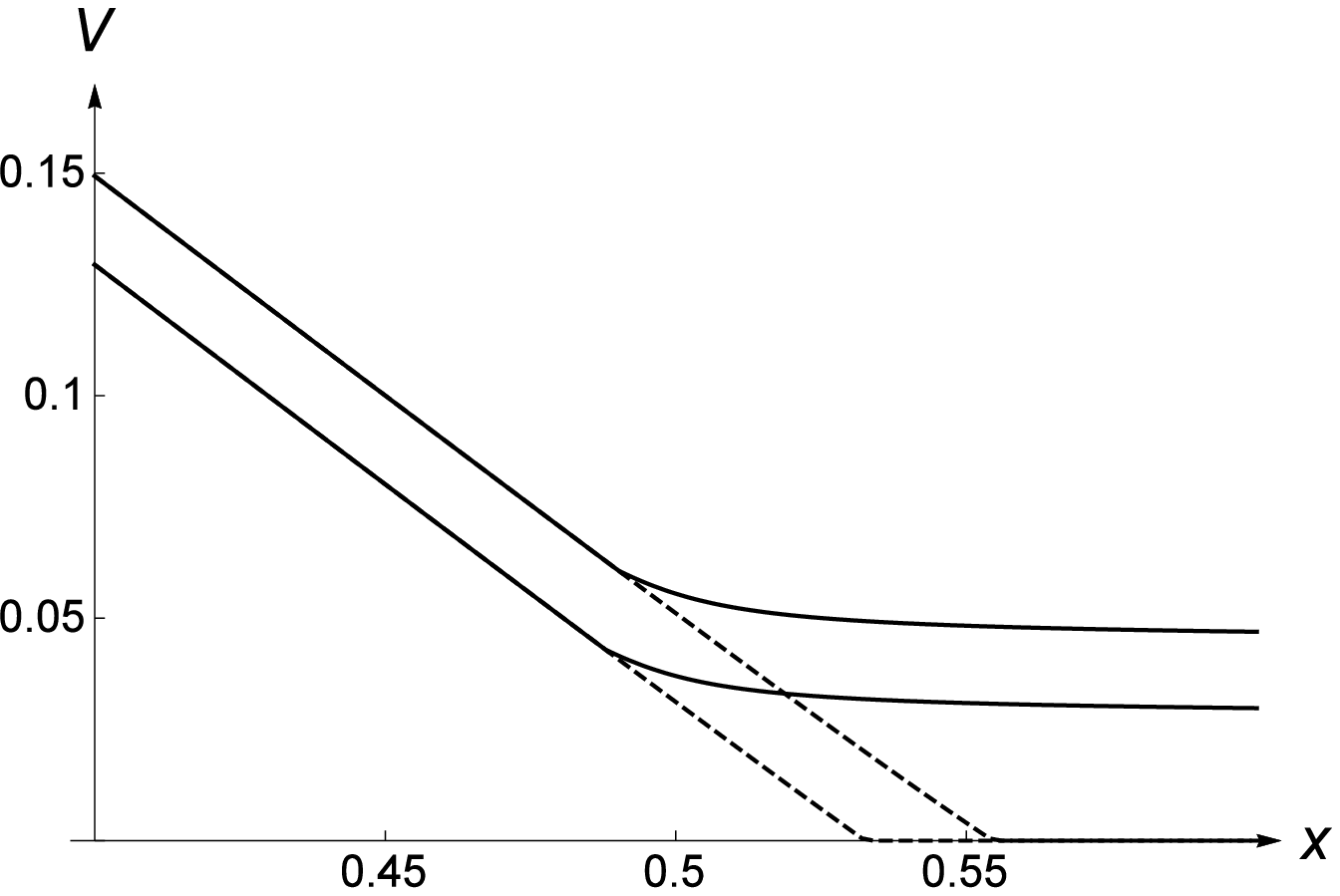}\label{fig4}
\end{center}

{\par \leftskip=1.6cm \rightskip=1.6cm \small \ni 
\textbf{Figure 4.} The value function $V^{1,E}(0,x)$ (solid) and payoff function $G^{1,E}(x)=(V^{1,L}(x)-x-c)^+$ (dashed) for the long-short strategy, plotted over $x$ for different transaction costs: $c = 0.01$ (upper), and $c=0.02$ (lower).   The parameters are: $T=T'=1$ year, $r=0.01, c=0.01, \theta=0.54, \mu=16, \zeta=0.16, \gamma^{1,L}=0.5545.$
\par}  
\end{figure}

 \section{Optimal short-long strategy}\label{sl}

The analysis of short-long strategy is completely symmetric to long-short one from the previous section. However we provide the main results for the sake of completeness and more importantly, we will use the solution to exit problem in the next section for chooser strategy.

We formulate the problem sequentially, assuming first that there is open short position in the spread which we want to liquidate optimally before $T'$
 \begin{equation} \label{sl-problem-1}  
V^{2,L}(x)=\inf \limits_{0\le\zeta\le T'}\EE \left[\e^{-r\zeta} (X^x_\zeta+c)\right],
 \end{equation}
for $x\in \I$ and the trader's optimal entry timing problem is given by
  \begin{equation} \label{sl-problem-2}  
V^{2,E}(t,x)=\sup \limits_{t\le\tau\le T}\EE \left[\e^{-r(\tau-t)}(X^{t,x}_\tau\m c\m V^{2,L}(X^{t,x}_\tau))^+\right],
 \end{equation}
for $(t,x)\in [0,T)\times \I$ as at time $\tau$ we receive $X^x_\tau$, pay $c$ and get the short position with the value $V^{2,L}(X^{t,x}_\tau)$. First we discuss the trivial cases.

\begin{prop} If $x_*\le a$, then it is optimal to wait until the end and exit from the position at $t=T'$. If $x_*\ge b$,  then it is optimal to liquidate the position at $t=0$. \end{prop}

Now under assumption $x_*\in\I$, we provide the
solution to the problem \eqref{sl-problem-1}.

\begin{theorem}\label{th:3}
The optimal stopping time  for \eqref{sl-problem-1} is given   by
\begin{align}  
\zeta_*^{2,L}=\inf\ \{\, 0\leq s\leq T': X^{t,x}_{s}\le b^{2,L}(s) \, \}.
 \end{align}
 The function   $b^{2,L}:[0,T]\rightarrow \R$ is the optimal exit boundary corresponding to  \eqref{sl-problem-1}, and it can be characterized as the unique solution to a nonlinear integral equation of Volterra-type
\begin{align}  
b^{2,L}(t)\p c=\e^{-r(T'-t)}m(T'\m t,b^{2,L}(t)) +\int_t^{T'} \cK^{2,L}(t,u,b^{2,L}(t),b^{2,L}(u))du,
\end{align}
for $t\in[0,T']$ in the class of continuous increasing functions $t\mapsto b^{2,L}(t)$ with
$b^{2,L}(T')=x_*$ and where the function
\begin{align}  
\cK^{2,L}(t,u,x,z):=-\e^{-r(u-t)}\EE \big[H^{2,L}(X^{t,x}_u) 1_{\{X^{t,x}_u \le z\}}\big],
 \end{align}
 for $u\ge t>0$ and $x,z\in\I$ with
 \begin{align} 
 H^{2,L}(x):=-(\mu+r)x+\mu\theta-rc\,, \label{H2Lx}
  \end{align}
for $x\in\I$. The value function $V^{2,L}$ in \eqref{sl-problem-1} admits the  representation
\begin{align} 
V^{2,L}(t,x;T')=\e^{-r(T'-t)}m(T'\m t,x) +\int_t^{T'}\cK^{2,L}(t,u,x,b^{2,L}(u))du\,,
\end{align}
for $t\in[0,T']$ and $x\in \I$.
\end{theorem}

Let us define now the threshold $\gamma^{2,L}$ as the unique root of the equation $x-c-V^{2,L}(x)=0$ so that we can rewrite the payoff
of the problem \eqref{sl-problem-2} as $G^{2,E}(x)=(x-c-V^{2,L}(x))1_{\{x>\gamma^{2,L}\}}$. As in the previous section, in order to exclude degenerate cases we assume 
that $\gamma^{2,L}\vee x^*\in \I$. We then have the following result for the entry problem.

\begin{theorem}\label{th:4}
The optimal stopping time  for \eqref{sl-problem-2} is given   by
\begin{align}  
\tau_*^{2,E}=\inf\ \{\, t\leq s\leq T: X^{t,x}_{s}\ge b^{2,E}(s) \, \}.
 \end{align}
The optimal entry boundary $b^{2,E}:[0,T]\rightarrow \R$ can be characterized as the unique solution to the recursive integral equation
\begin{align}\label{sl-th-1-tc}  
b^{2,E}(t)\m c \m V^{2,L}(b^{2,E}(t))=&e^{-r(T-t)}\EE [G^{2,E}(X^{t,b^{2,E}(t)}_{T})]\\
&+\int_t^{T} \cK^{2,E}(t,u,b^{2,E}(t),b^{2,E}(u))du,\notag
\end{align}
for $t\in[0,T]$ in the class of continuous decreasing functions with $b^{2,E}(T-)=\gamma^{2,L}\vee x^*$, where
 the   function
\begin{align}  
\cK^{2,E}(t,u,x,z):=-\e^{-r(u-t)}\EE \big[H^{2,E}(X^{t,x}_u) 1_{\{X^{t,x}_u \ge z\}}\big],
 \end{align}
 for $u\ge t>0$ and $x,z\in\I$ with $H^{2,E}(x):=-(\mu+r)x+\mu\theta+rc=H^{1,L}(x)$.
 The value function $V^{2,E}$ can be represented as
\begin{align}\label{sl-th-2-tc}  
V^{2,E}(t,x)=e^{-r(T-t)}\EE [G^{2,E}(X^{t,x}_{T})]+ \int_t^{T}\cK^{2,E}(t,u,x,b^{2,E}(u))du,
\end{align}
for $t\in[0,T]$ and $x\in \I$.
\end{theorem}

\section{Chooser strategy}\label{o}

Now we aggregate the long-short and short-long strategies into one strategy whereby the trader can choose upon entry whether to go long or short in $X$ at or before the deadline $T$. As such, the strategy is called the chooser strategy.  In other words, the trader is not pre-committed to the type of trade prior to market entry, and thus, has a   chooser option in addition to the timing options that are embedded in her trading problem. Clearly, this added flexibility should increase the   expected profit from trading in $X$. After entry, the trader has a separate deadline  $T'$ to liquidate the position.

As in Sections 3 and 4, we tackle the trading problem sequentially.
 Once the trader enters into the position, she/he solves one of the optimal liquidation problems and both of them were already discussed in Theorems \ref{th:1} and \ref{th:3}.
Therefore it remains to analyze the optimal entry problem, which  is  given by
 \begin{equation} \label{0-problem-ch}  
V^{0,E}(t,x)=\sup \limits_{t\le\tau\le T}\EE \left[\e^{-r(\tau-t)}G^{0,E}(X^{t,x}_\tau)\right],
 \end{equation}
for $t\in[0,T)$ and $x\in \I$, where the payoff function $G^{0,E}$ reads
 \begin{equation} \label{0-payoff-ch}  
G^{0,E}(x)=\max\left\{V^{1,L}(x)\m x\m c,x\m c\m V^{2,L}(x),0\right\},
 \end{equation}
  for $x\in \I$. The payoff function $G^{0,E}$ shows that at entry time the trader maximizes his value and chooses the best option, i.e. whether to go long or short the spread, or not to enter at all. The latter may happen to be optimal due to the presence of trading costs. As in the previous sections, the gain function is time-independent. We recall that in order exclude the trivial situations for both exit problems we assume that $a<x_*\le x^*<b$.

By the nature of the value functions, we know that $V^{1,L}(t,x) =x-c$ for $x\ge b^{1,L}(0)$ and $V^{2,L}(t,x)=x+c$ for $x\le b^{2,L}(0)$ for any $t\in[0,T)$.
 In addition,  since $V^{1,L}$ and $V^{2,L}$ are convex and concave, respectively, we have   $V^{1,L}_x \le 1$ and $V^{2,L}_x \le1$. Hence,
 the function $V^{1,L}(x)\m x\m c$ is decreasing for $x< b^{1,L}(0)$, and $x\m c\m V^{2,L}(x)$ is increasing $x> b^{2,L}(0)$. In turn, we can  conclude that there exists a constant threshold $m\in (b^{2,L}(0),b^{1,L}(0))$  such that
 \begin{equation} \label{0-payoff-2-ch}  
G^{0,E}(x)=\left((V^{1,L}(x)\m x\m c ) 1_{\{x\le m\}} + (x\m c\m V^{2,L}(x))1_{\{x> m\}}\right)^+\,,
 \end{equation}
for $x\in \I$. The important matter is the sign $V^{1,L}(m)- m- c=m- c- V^{2,L}(m)$. Using the definitions of $\gamma^{1,L}$ and $\gamma^{2,L}$,
the function $G^{0,E}$ can be then rewritten as
  \begin{equation} \label{0-payoff-2-ch}  
G^{0,E}(x)=(V^{1,L}(x)\m x\m c ) 1_{\{x\le m\wedge \gamma^{1,L}\}} + (x\m c\m V^{2,L}(x))1_{\{x> m\vee \gamma^{2,L}\}},
 \end{equation}
for $x\in \I$. Then there are two possibilities: (i) $\gamma^{1,L}<m<\gamma^{2,L}$ or (ii) $\gamma^{2,L}<m<\gamma^{1,L}$. Also, we note that the function $G^{0,E}$ is convex. 
As usual in this paper we assume that $m,\gamma^{1,L},\gamma^{2,L}\in \I$ to avoid degenerate cases.

As usual, define the continuation and entry regions, respectively, by
\begin{align} \label{C-ch}  
&\cC^{0,E}= \{\, (t,x)\in[0,T)\! \times\! \I:V^{0,E}(t,x)>G^{0,E}(x)\, \} ,\\
 \label{D-ch}&\cD^{0,E}= \{\, (t,x)\in[0,T)\! \times\! \I:V^{0,E}(t,x)=G^{0,E}(x)\, \}.
 \end{align}
Then the optimal entry time in \eqref{0-problem-ch} is given by
\begin{align} \label{OST-ch}  
\tau^{0,E}=\inf\ \{\ t\leq s\leq T:(s,X^{t,x}_{s})\in\cD^{0,E}\ \}.
 \end{align}

 We now provide the main result of this section.
 \begin{theorem}\label{th:ch}
There exists  a pair of boundaries $(b^{0,E},\wt{b}^{0,E})$ such that
\begin{align}   \label{tau-0E}
&\tau_{b}=\inf\ \{\ t\leq s\leq T: X^{t,x}_{s}\le b^{0,E}(s)\;\text{or}\; X^{t,x}_{s}\ge \wt{b}^{0,E}(s)\ \}
 \end{align}
is optimal in \eqref{0-problem-ch}. This pair can be characterized as the unique solution to a system of coupled integral equations
\begin{align}\label{th-1-ch}  
V^{1,L}(t,b^{0,E}(t))\m b^{0,E}(t)\m c=&e^{-r(T-t)}\EE [G^{0,E}(X^{t,b^{0,E}(t)}_{T})]\\
&+\int_t^{T} \cK^{0,E}(t,u,b^{0,E}(t),b^{0,E}(u),\wt{b}^{0,E}(u))du,\notag\\
\label{th-2-ch} \wt{b}^{0,E}(t)\m c\m V^{2,L}(t,\wt{b}^{0,E}(t))=&e^{-r(T-t)}\EE [G^{0,E}(X^{t,\wt{b}^{0,E}(t)}_{T})]\\
&+\int_t^{T} \cK^{0,E}(t,u,\wt{b}^{0,E}(t),b^{0,E}(u),\wt{b}^{0,E}(u))du,\notag
\end{align}
for $t\in[0,T]$ in the class of continuous increasing functions $b^{0,E}$ with $b^{0,E}(T)=m\wedge\gamma^{1,L}\wedge x^*$
and continuous decreasing functions $\wt{b}^{0,E}$ with $\wt{b}^{0,E}(T)=m\vee\gamma^{2,L}\vee x^*$.

 The value function $V^{0,E}$ can be represented as
\begin{align}\label{th-3-ch}  
V^{0,E}(t,x)=e^{-r(T-t)}\EE [G^{0,E}(X^{t,x}_{T})]+ \int_t^{T}\cK^{0,E}(t,u,x,b^{0,E}(u),\wt{b}^{0,E}(u))du,
\end{align}
for $t\in[0,T]$ and $x\in \I$, where the  function $\cK^{0,E}$ defined as
\begin{align}
\label{L-ch}
&\cK^{0,E}(t,u,x,z,\wt{z}):=-\e^{-r(u-t)}\EE \big[H^{0,E}( X^{t,x}_u) 1_{\{X^{t,x}_u \le z\;\text{or}\;X^{t,x}_u \ge \wt{z}\}}\big]\\
&=-\e^{-r(u-t)}\left(\int_{-\infty\vee a}^z H^{0,E}(\wt{x})\,p(\wt{x};u,x,t)\,d\wt{x}+\int_{\wt{z}}^{\infty\wedge b} H^{0,E}(\wt{x})\,p(\wt{x};u,x,t)\,d\wt{x}\right),\notag
 \end{align}
 for $u\ge t\ge 0$ and $x,z,\wt{z}\in\I$.
\end{theorem}

\begin{figure}[t]
\begin{center}
\includegraphics[scale=0.7]{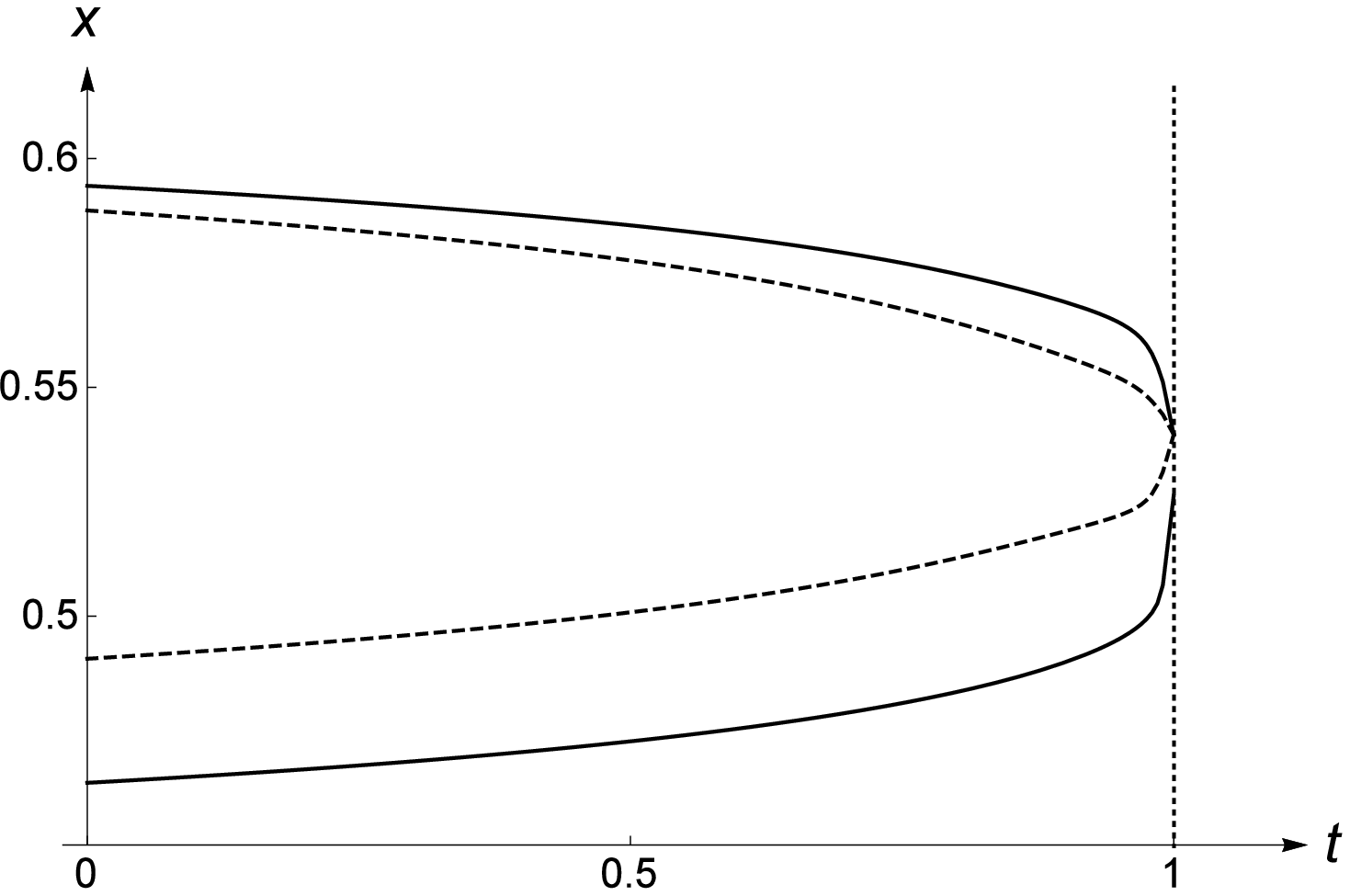}
\end{center}

{\par \leftskip=1.6cm \rightskip=1.6cm \small \ni 
\textbf{Figure 5.} The optimal entry  boundaries corresponding to the  long-short strategy (lower dashed) and short-long strategy (upper dashed) are enclosed by the optimal lower and upper entry boundaries $(b^{0,E},\wt{b}^{0,E})$ (solid)  associated with the chooser strategy. The boundary pair $(b^{0,E},\wt{b}^{0,E})$ is  numerically solved  from the  system of integral equations \eqref{th-1-ch}-\eqref{th-2-ch} under the OU model. The parameters are:  $T=T'=1$ year, $r=0.01, c=0.01, \theta=0.54, \mu=16, \sigma=0.16.$  A time discretization with  $500$ steps for the interval $[0,T]$ is used.
\par} 
\end{figure}
\vs{10pt}

Before presenting the proof of Theorem \ref{th:ch}, let us illustrate and discuss the properties of  the optimal trading boundaries. We  numerically solve for the optimal  boundary pair $(b^{0,E},\wt{b}^{0,E})$  from the  system of integral equations \eqref{th-1-ch}-\eqref{th-2-ch} under the OU model.   In Figure 5, we observe that the optimal entry  boundaries corresponding to the  long-short strategy   and short-long strategy  are enclosed by the optimal lower and upper entry boundary pair $(b^{0,E},\wt{b}^{0,E})$   associated with the chooser strategy.   In other words, with the option to select the first (long/short) position, it is optimal for the trader to wait longer for a better entry price. To see this, let us take some $t'\in[0,T)$ and $x'\le m\wedge \gamma^{1,L}$ so that $G^{0,E}(x')=G^{1,E}(x')$.  Since   $V^{0,E}(t',x')\ge V^{1,E}(t',x')$, if  $V^{1,E}(t',x')-G^{1,E}(t',x')>0$, i.e. it is not optimal to enter into the long position, then $V^{0,E}(t',x')-G^{0,E}(t',x')\ge V^{1,E}(t',x')-G^{1,E}(t',x')>0$, which means $(t',x')\in \cC^{0,E}$.
As a consequence, we conclude that $b^{0,E}(t)\le b^{1,E}(t)$ for any $t\in [0,T)$.  Using similar arguments, we can show that $\wt{b}^{0,E}(t)\ge b^{2,E}(t)$ for any $t\in[0,T)$.

\begin{proof} First, to obtain some properties of the entry region, we use Ito-Tanaka's formula to get 
\begin{align} \label{Ito-ch}  
 \EE \left[\e^{-r(\tau-t)}G^{0,E}(X^{t,x}_\tau)\right]=&\;G^{0,E}(x)+\EE \left[\int_t^\tau \e^{-r(s-t)}H^{0,E}(X^{t,x}_s)ds\right]\\
 &+\EE \left[\int_t^\tau \e^{-r(s-t)}(1-V^{1,L}_x(X^{t,x}_s))d\ell^{m\wedge\gamma^{1,L}}_s\right]\nonumber\\
 &+\EE \left[\int_t^\tau \e^{-r(s-t)}(1-V^{2,L}_x(X^{t,x}_s))d\ell^{m\vee\gamma^{2,L}}_s\right],\nonumber
 \end{align}
for $t\in[0,T)$, $x\in\I$, any stopping time $\tau$ of process $X$. Here,  the function $H^{0,E}$ is defined as  $H^{0,E}(x):=(\L_X G^{0,E} \m rG^{0,E})(x)$ for $x\in\I$ and equals
\begin{align} \label{H-ch}  
H^{0,E}(x)=\left((\mu\p r)x-\mu\theta+rc\right)(1_{\{x <m\wedge \gamma^{1,L}\}}-1_{\{x >m\vee \gamma^{2,L}\}}).
\end{align}

The function $H^{0,E}$ is negative when $x\in(-\infty,m\wedge\gamma^{1,L}\wedge x^*)\cup(m\vee \gamma^{2,L}\vee x^*,\infty)$. It is not optimal to enter into the position when $m\wedge \gamma^{1,L}\wedge x^*<X_s<m\vee \gamma^{2,L}\vee x^*$
as $H^{0,E}$ is positive there and the local time term is always non-negative. 
Also, near $T$ it is optimal to enter at once when $X_s<m\wedge\gamma^{1,L}\wedge x^*$ or $X_s>m\vee \gamma^{2,L}\vee x^*$ due to lack of time to compensate the negative $H^{0,E}$. This gives us the terminal condition of the exercise boundaries (see below) at $T$.
We also note that if $a=-\infty$ and $b=\infty$, the equation \eqref{Ito-ch} shows that the entry region is non-empty for all $t\in[0,T)$, as for large negative or positive $x$ the integrand $H^{0,E}$
is very negative and thus it is optimal to enter immediately due to the presence of the finite deadline $T$.
\vs{2pt}

Since the payoff function $G^{0,E}$ is time-homogenous, we have that the entry region $\cD^{0,E}$ is  right-connected.
Next, we show that  $\cD^{0,E}$ is both down- and up-connected. We prove only that it is down-connected, for up-connectedness the same arguments can be applied.
Let us take $t>0$ and $x<y<m\wedge\gamma^{1,L}\wedge x^*$ such that $(t,y)\in \cD^{0,E}$.
Then, by right-connectedness of the entry region, we have that $(s,y)\in \cD^{0,E}$ as well for any $s>t$. If we now run the process $(s,X_s)_{s\ge t}$ from $(t,x)$, we cannot hit the level $m\wedge\gamma^{1,L}\wedge x^*$ before entry (as $x<y$), thus the local time term in \eqref{Ito-ch} is 0 and integrand $H^{0,E}$ is negative before $\tau^{E}_*$.
Therefore it is optimal to enter at $(t,x)$ and we obtain down-connectedness of the entry region $\cD^{0,E}$.

Hence there exists a pair of optimal entry boundaries $(b^{0,E},\wt{b}^{0,E})$ on $[0,T)$ such that
$\tau_{b}$ defined in \eqref{tau-0E} is optimal in \eqref{0-problem-ch} and $-\infty<b^{0,E}(t)<m\wedge\gamma^{1,L}\wedge x^*<m\vee\gamma^{2,L}\vee x^*<\wt{b}^{0,E}(t)<\infty$ for $t\in[0,T)$.
From arguments above, we also have that $b^{0,E}(T-)=m\wedge\gamma^{1,L}\wedge x^*$ and
$\wt{b}^{0,E}(T-)=m\vee\gamma^{2,L}\vee x^*$.
Moreover, right-connectedness of $\cD^{0,E}$ implies that $b^{0,E}$ is increasing and $\wt{b}^{0,E}$ is decreasing on $[0,T)$. 
\vs{6pt}

 Standard arguments based on strong Markov property show that the value function $V^{0,E}$ and boundaries $(b^{0,E},\wt{b}^{0,E})$ solve the following free-boundary problem:
\begin{align} \label{PDE-ch}  
&V^{0,E}_t \p\L_X V^{0,E} \m rV^{0,E}=0 &\hs{-30pt}\text{in}\;  \cC^{0,E},\\
\label{IS-ch-1}&V^{0,E}(t,b^{0,E}(t))=V^{1,L}(t,b^{0,E}(t))\m b^{0,E}(t)\m c &\hs{-30pt}\text{for}\; t\in[0,T),\\
\label{IS-ch-2}&V^{0,E}(t,\wt{b}^{0,E}(t))=\wt{b}^{0,E}(t)\m c\m V^{2,L}(t,\wt{b}^{0,E}(t)) &\hs{-30pt}\text{for}\; t\in[0,T),\\
\label{SF-ch-1}&V^{0,E}_x(t,b^{0,E}(t))=V^{1,L}_x(t,b^{0,E}(t))\m 1&\hs{-30pt}\text{for}\; t\in[0,T), \\
\label{SF-ch-2}&V^{0,E}_x(t,\wt{b}^{0,E}(t))=1\m V^{2,L}_x(t,\wt{b}^{0,E}(t))&\hs{-30pt}\text{for}\; t\in[0,T), \\
\label{FBP1-ch}&V^{0,E}(t,x)>G^{0,E}(x) &\hs{-30pt}\text{in}\; \cC^{0,E},\\
\label{FBP2-ch}&V^{0,E}(t,x)=G^{0,E}(x) &\hs{-30pt}\text{in}\; \cD^{0,E},
\end{align}
where the continuation set $\cC^{0,E}$ and the entry region $\cD^{0,E}$ are given by
\begin{align} \label{CC-ch}  
&\cC^{0,E}= \{\, (t,x)\in[0,T)\! \times\! \I: b^{0,E}(t)<x<\wt{b}^{0,E}(t)\, \},\\ 
 \label{DD-ch}&\cD^{0,E}= \{\, (t,x)\in[0,T)\! \times\! \I:x\le b^{0,E}(t)\;\text{or}\;x\ge \wt{b}^{0,E}(t)\, \}.
 \end{align}
The following properties of $V^{0,E}$ and $(b^{0,E},\wt{b}^{0,E})$ also hold:
\begin{align} \label{Prop-1-ch}  
&V^{0,E}\;\text{is continuous on}\; [0,T]\times\I,\\
\label{Prop-3-ch}&x\mapsto V^{0,E}(t,x)\;\text{is convex on $\I$ for each $t\in[0,T]$},\\
\label{Prop-4-ch}&t\mapsto V^{0,E}(t,x)\;\text{is decreasing on $[0,T]$ for each $x\in\I$},\\
\label{Prop-5-ch}&b^{0,E}\;\text{and}\;\wt{b}^{0,E}\;\text{are continuous on}\;[0,T].
\end{align}

We now apply the local time-space formula (\cite{Peskir2005a}) to the process  $\e^{-r(s-t)}V^{0,E}(s,X^{t,x}_s)$
 along with \eqref{PDE-ch}, the definition of $H^{0,E}$, the smooth-fit properties \eqref{SF-ch-1} and \eqref{SF-ch-2} to get
 \begin{align} \label{LTSF-ch}   
\e^{-r(s-t)}&V^{0,E}(s,X^{t,x}_s)\\
=\;&V^{0,E}(t,x)+M_s\nonumber\\
 &+ \int_t^{s}\e^{-r(u-t)}\left(V^{0,E}_t \p \L_X V^{0,E}\m rV^{0,E}\right)(u,X^{t,x}_u)du\nonumber\\
 &+\frac{1}{2}\int_t^{s}\e^{-r(u-t)}\left(V^{0,E}_x (u,X^{t,x}_u +)-V^{0,E}_x (u,X^{t,x}_u -)\right)d\ell^{b^{0,E}}_u\nonumber\\
 &+\frac{1}{2}\int_t^{s}\e^{-r(u-t)}\left(V^{0,E}_x (u,X^{t,x}_u +)-V^{0,E}_x (u,X^{t,x}_u -)\right)d\ell^{\wt{b}^{0,E}}_u\nonumber\\
  =\;&V^{0,E}(t,x)+M_s \\&+\int_t^{s}\e^{-r(u-t)}H^{0,E}(X^{t,x}_u) 1_{\{X^{t,x}_u \le b^{0,E}(u)\;\text{or}\;X^{t,x}_u \ge \wt{b}^{0,E}(u)\}}du,\nonumber
  \end{align}
where    $M=(M_s)_{s\ge t}$ is the martingale part, $(\ell^{b^{0,E}}_s)_{s\ge t}$ and $(\ell^{\wt{b}^{0,E}}_s)_{s\ge t}$ are the local time processes of $X^x$ at the boundaries $b^{0,E}$ and
$\wt{b}^{0,E}$, respectively. Now letting $s=T$, taking the expectation $\EE$, using the optional sampling theorem, rearranging terms and noting that $V^{0,E}(T,\cdot)=G^{0,E}(\cdot)$,
we obtain \eqref{th-3-ch}. Then by inserting $x=b^{0,E}(t)$ and $x=\wt{b}^{0,E}(t)$ into \eqref{th-3-ch}, and recalling the continuous pasting properties \eqref{IS-ch-1}+\eqref{IS-ch-2}, we arrive at  the system of coupled integral equations \eqref{th-1-ch}-\eqref{th-2-ch}. \end{proof}

\newpage
  Figure 6 shows that  the  map $T\mapsto V^{0,E}(0,\theta;T)$ corresponding to  the chooser strategy, evaluated at $x = \theta$, is an increasing function of the deadline $T$. However, the slopes appears to be decreasing rapidly as $T$ increases, indicating a reduced benefit of a longer trading horizon. For every $T$,   $V^{0,E}(0,\theta;T)$   dominates the value function $V^{1,E}(0, \theta;T)$   for the long-short strategy.  This difference in value, which is very substantial in this example, can be viewed as the premium associated with the chooser option in $V^{0,E}(0,\theta;T)$.

   In Figure 7, we compare the value functions for the optimal  entry problems with and without the chooser option.  The value function $V^{0,E}(0,x)$ for the chooser strategy (solid) dominates the payoff function $G^{0,E}(x)$ (see \eqref{0-payoff-ch}) that is V-shaped and plotted in dotted line, and the two coincide for sufficiently large and small $x$. In comparison, $V^{1,E}(x)$ (dashed) for the long-short strategy only dominates the payoff function $G^{0,E}(x)$ for small $x$ on the left. This means that for large $x$, immediately exercising the chooser option and capturing the payoff $G^{0,E}(x)$ is better than optimally timing to enter the market with the long-short strategy.  Moreover, the value function $V^{0,E}(0,x)$ (chooser strategy) is higher than $V^{1,E}(x)$ (long-short) for all $x$, and they coincide for sufficiently small $x$ when immediate market entry (with a long position) is optimal for the chooser strategy.

\begin{figure}[h]
\begin{center}
\includegraphics[scale=0.8]{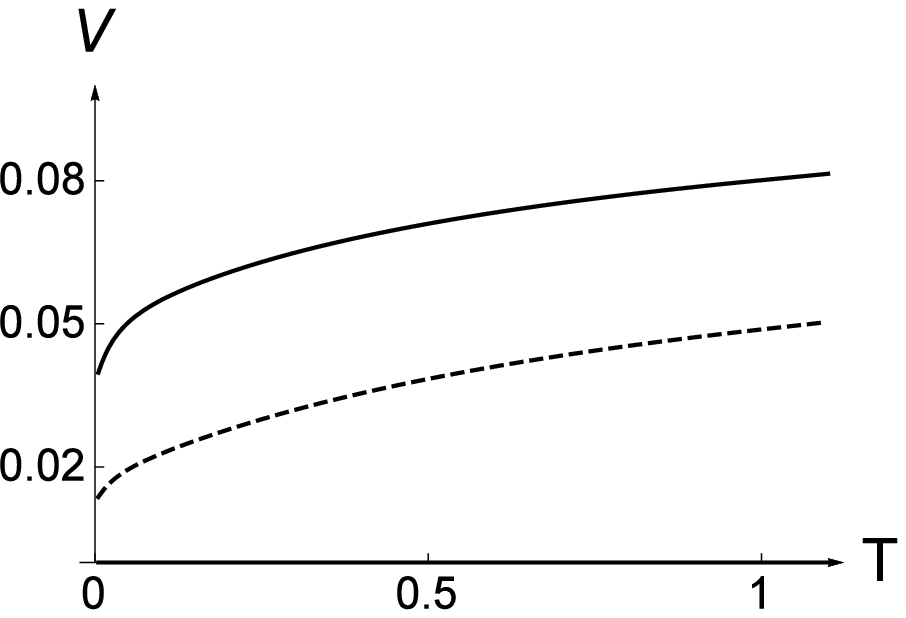}
\end{center}

{\par \leftskip=1.6cm \rightskip=1.6cm \small \ni  
\textbf{Figure 6.} The value function $V^{0,E}(0,x;T)$ associated with  the chooser strategy (solid), evaluated at $x = \theta$ and plotted  as a function of the deadline $T$ years.  For every $T$,  $V^{0,E}(0,\theta;T)$   dominates the value function $V^{1,E}(0, \theta;T)$ (dashed) for the long-short strategy. The parameters are: $T'=1$ year, $r=0.01, c=0.01, \theta=0.54, \mu=16, \zeta=0.16.$
\par} 
 \end{figure}

\begin{figure}[h]
\begin{center}
\includegraphics[scale=0.8]{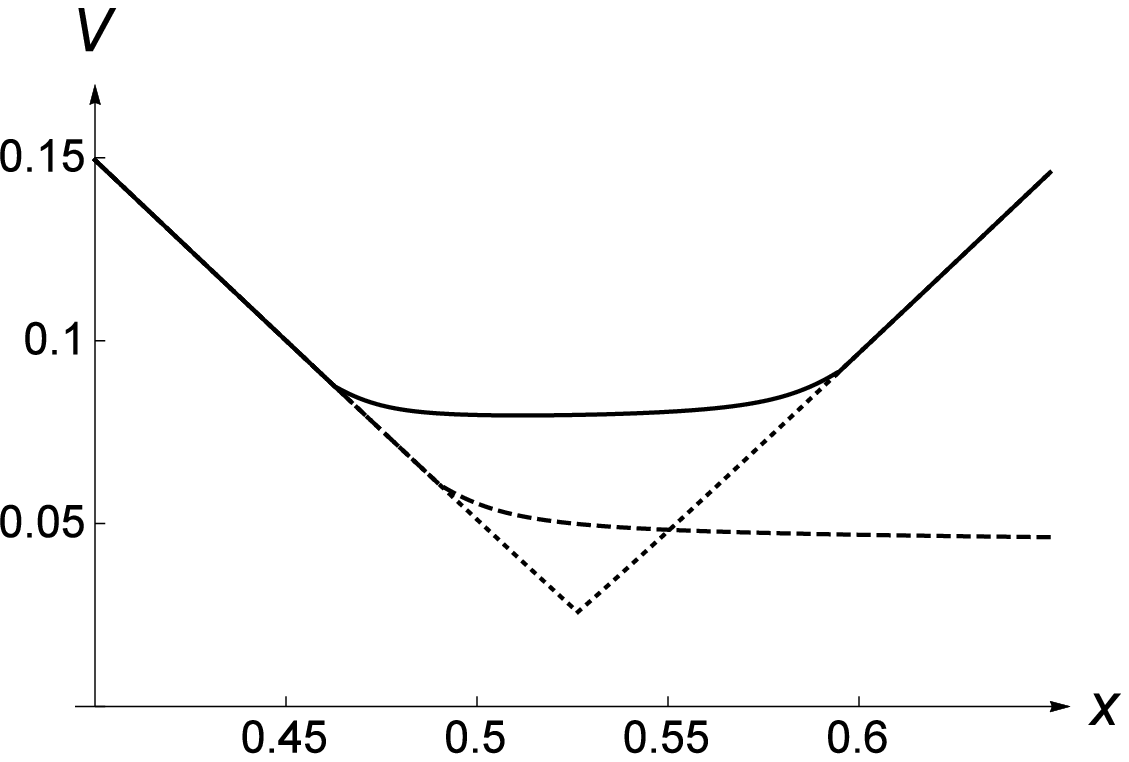}
\end{center}
{\par \leftskip=1.6cm \rightskip=1.6cm \small \ni 
\textbf{Figure 7.} The value function $V^{0,E}(0,x)$ (solid) as a function of $x$ corresponding to the optimal  entry problem with the chooser strategy versus   $V^{1,E}(x)$ (dashed) for the long-short strategy. The dotted line represents the payoff function $G^{0,E}(x)$ (see \eqref{0-payoff-ch}) associated with the value function $V^{0,E}(0,x)$. The parameters are: $T=T'=1$ year, $r=0.01, c=0.01, \theta=0.54, \mu=16, \zeta=0.16, \gamma^{1,L}=0.5545.$
\par}  
\end{figure}
 \clearpage

%
\begin{small}
\begin{spacing}{0.7}
\bibliographystyle{apa}
\bibliography{mybib2}
\end{spacing}
\end{small}

\end{document}